\documentclass[a4paper,11pt]{article}
\usepackage{amsmath, amsthm, amssymb}
\usepackage{tabularx}
\usepackage{algorithm}
\usepackage{algpseudocode}
\usepackage{tikz}
\usetikzlibrary {arrows.meta,automata,positioning}
\usepackage{hyperref}
\hypersetup{breaklinks=True}
\usepackage{xurl}
\usepackage{geometry}
\usepackage{comment}

\makeatletter
\newcommand{\multiline}[1]{%
  \begin{tabularx}{\dimexpr\linewidth-\ALG@thistlm}[t]{@{}X@{}}
    #1
  \end{tabularx}
}
\makeatother

\newtheorem{theorem}{Theorem}[section]

\newtheorem{lemma}[theorem]{Lemma}

\theoremstyle{definition}
\newtheorem{definition}[theorem]{Definition}
\newtheorem{example}[theorem]{Example}

\title{Optimal Multidimensional Convolutional Codes} 

\author{Z. Abreu, J. Lieb, R. Pinto and R. Sim\~oes}

\begin{document}
\maketitle
\begin{abstract}
In this paper, we analyze $m$-dimensional ($m$D) convolutional codes with finite support, viewed as a natural generalization of one-dimensional (1D) convolutional codes to higher dimensions. An $m$D convolutional code with finite support consists of codewords with compact support indexed in $\mathbb{N}^m$ and taking values in $\mathbb{F}_{q}[z_1,\ldots,z_m]^n$, where $\mathbb{F}_{q}$ is a finite field with $q$ elements. We recall a natural upper bound on the free distance of an $m$D convolutional code with rate $k/n$ and degree~$\delta$, called $m$D generalized Singleton bound. Codes that attain this bound are called maximum distance separable (MDS) $m$D convolutional codes. As our main result,
 we develop new constructions of MDS $m$D convolutional codes based on superregularity of certain matrices. Our results include the construction of new families of MDS $mD$  convolutional codes of rate $1/n$, relying on generator matrices with specific row degree conditions. These constructions significantly expand the set of known constructions of MDS $m$D convolutional codes.
\end{abstract}

\section{Introduction}

A one-dimensional ($1$D) convolutional code can be described as an \( \mathbb{F}_{q}[z] \)-submodule of \( \mathbb{F}_{q}[z]^n \), where \( \mathbb{F}[z] \) denotes the polynomial ring in a single indeterminate over the field with $q$ elements \( \mathbb{F}_{q} \). One significant advantage of convolutional codes is their enhanced error correction capabilities compared to block codes, especially in scenarios where data is transmitted in a continuous stream. The structure of convolutional codes allows the detection and correction of errors that may occur across multiple transmitted symbols, rather than being limited to fixed blocks. Notable contributions to the theory of convolutional codes were made by Forney, see \cite{Forney1970a,Forney1973,Forney1974}. The search for convolutional codes with optimal encoding and decoding properties remains an active area of research. An excellent introduction to convolutional codes can be found in the books \cite{JohannessonZigangirov1999,LiebPintoRosenthal2021,LinCostello1983}.

Multidimensional ($m$D) convolutional codes extend the concept of convolutional codes to polynomial rings with multiple variables. Consider the polynomial ring \( R = \mathbb{F}_{q}[z_1, \dots, z_m] \) in \( m \) indeterminates over \( \mathbb{F}_{q} \). An \( m \)-dimensional convolutional code of length \( n \) is defined as an \( R \)-submodule of \( R^n \).

$m$D codes offer significant advantages in the transmission of multidimensional data. For instance, $2$D convolutional codes are suited for applications such as transmitting images and videos as $2$D data. With the advancement of higher dimensional codes, we are confident that these applications will find more use in the future.

While one-dimensional (1D) convolutional codes have been widely investigated, two-dimensional (2D) convolutional codes have received comparatively less attention in the literature. The concept of 2D convolutional codes was first introduced by Fornasini and Valcher in \cite{ValcherFornasini1994,FornasiniValcher1994}. Subsequently, the authors in \cite{ClimentNappPereaPinto2016} established an upper bound for the free distance of 2D convolutional codes and proposed optimal code constructions. Further research on the construction and characterization of 2D convolutional codes can be found in \cite{AlmeidaNappPinto2018,AlmeidaNappPinto2015,ClimentNappPereaPinto2012}. 

Despite these advances, there are still relatively few explicit constructions of maximum distance separable (MDS) 2D convolutional codes. In \cite{ClimentNappPereaPinto2012}, a family of 2D convolutional codes with rate $1/n$ and degree $\delta$ was proposed for
\[
n \geq \frac{(\delta + 1)(\delta + 2)}{2}.
\]
Later, \cite{AlmeidaNappPinto2015} introduced constructions of 2D MDS convolutional codes with rate $1/n$ and degree $\delta$ for $\delta \leq 2$ and $n \geq \delta + 1$. In \cite{ClimentNappPereaPinto2016}, the existence of MDS 2D convolutional codes with rate $k/n$ and degree $\delta$ was proven through a constructive procedure under the following conditions:
\begin{align*}
n \geq k \frac{\left(\left\lfloor \frac{\delta}{k} \right\rfloor + 2 \right) \left(\left\lfloor \frac{\delta}{k} \right\rfloor + 3 \right)}{2}, \quad \text{if } k \nmid \delta,
\end{align*}
and
\begin{align*}
n \geq k \frac{\left(\left\lfloor \frac{\delta}{k} \right\rfloor + 1 \right) \left(\left\lfloor \frac{\delta}{k} \right\rfloor + 2 \right)}{2}, \quad \text{if } k \mid \delta.
\end{align*}
More recently, \cite{AlmeidaNappPinto2018} presented additional constructions of 2D MDS convolutional codes with rate $k/n$ and degree $\delta$ satisfying
\[
n \geq k \left( \left\lfloor \frac{\delta}{k} \right\rfloor + 1 \right).
\]
Together, these works establish important theoretical foundations and constructive methods for the development of 2D convolutional codes with optimal distance properties.

Higher-dimensional convolutional codes have garnered even less research. $m$D convolutional codes were first introduced in \cite{Weiner1998,GluesingLuerssenRosenthalWeiner2000} and further examined in \cite{Charoenlarpnopparut2009,CharoenlarpnopparutTantaratana2004,Kitchens2002,Zerz2010,Lomadze2020}. Key distinctions exist between $1$D and $2$D convolutional codes, as well as between $2$D and $m$D codes, with $m \geq 3$, with these differences being thoroughly explored by Weiner in \cite{Weiner1998}. Very recently, in \cite{MDSMulti2025}, the authors established an upper bound on the distance of an $m$D convolutional code of rate $k/n$ and degree $\delta$, called the $m$D generalized Singleton bound. They also presented the first construction of 3D convolutional codes with finite support, of rate $1/n$ and degree $\delta \leq 2$, that achieve this bound and so have the maximum possible distance compared to any other 3D convolutional codes with the same rate and degree.

In this paper, we develop new constructions of MDS $m$D convolutional codes of rate $1/n$. These constructions rely on superregularity conditions on suitable matrices and extend known 1D results to the multidimensional setting. Our goal is to develop an new approach for constructing MDS convolutional codes in higher dimensions.

The outline of this paper is as follows. In Section 2, we present the main concepts of multidimensional convolutional codes and, for completeness, we also provide a proof of the upper bound on the distance of $m$-dimensional ($m$D) convolutional codes. In Section 3, we focus on $m$D convolutional codes of rate $1/n$ and degree $\delta$, and we present constructions of such codes that are maximum distance separable (MDS). In the last theorem, we then extend these to the more general case of rate $k/n$. Finally, Section 4 provides a brief summary of the results of the paper.

\section{Preliminaries}

In this section, we introduce the fundamental concepts of multidimensional convolutional codes. These codes generalize $1$D convolutional codes to polynomial rings in multiple variables. Let $R=\mathbb{F}_{q}[z_1,\dots, z_m ]$ denote the polynomial ring in $m$ variables with coefficients in $\mathbb{F}_{q}$.

\begin{definition}[\hspace{-0.1mm}\cite{Weiner1998}]
An \textbf{$m$D} (finite support) \textbf{convolutional code} $\mathcal{C}$ of rate $k/n$ is a free $R$-submodule of $R^{n} $ of rank $k$. 

A full row rank matrix $G(z_1,\dots, z_m) \in R^{k \times n}$, whose rows form a basis for $\mathcal{C}$, is called a \textbf{generator matrix}\index{Generator matrix} of $\mathcal{C}$ and therefore

{\small
\[
\begin{aligned}
\mathcal{C} &= \operatorname{Im}_{R}G(z_1, \dots, z_m)\\
&= \left\{ v(z_1, \dots, z_m) \in R^n :\right.\\
&\quad \left. v(z_1, \dots, z_m) = u(z_1, \dots, z_m)\, G(z_1, \dots, z_m) \text{ with } u(z_1, \dots, z_m) \in R^k \right\}.
\end{aligned}
\]
}
\end{definition}


\begin{definition}[\hspace{-0.1mm}\cite{Weiner1998}]
A square matrix $U(z_1, \dots, z_m) \in R^{k \times k}$ is \textbf{unimodular} if there is a matrix $V(z_1, \dots, z_m) \in R^{k \times k}$ such that $$U(z_1, \dots, z_m) \cdot V(z_1, \dots, z_m) = V (z_1, \dots, z_m)\cdot U(z_1, \dots, z_m) = I_{k}.$$
\end{definition}

A matrix $U(z_1, \dots, z_m)\in R^{k \times k}$ is unimodular if and only if $$
\det(U(z_1, \dots, z_m)) \in \mathbb{F}_q \setminus \{0\} \quad \text{\cite{Weiner1998}}.$$

Two full row rank matrices $G_1(z_1, \dots, z_m)$, $G_2(z_1, \dots, z_m) \in R^{k \times n}$ are said to be \textbf{equivalent generator matrices} if they are generator matrices of the same $m$D convolutional code, i.e., if $$\operatorname{Im}_R G_1(z_1, \dots, z_m)= \operatorname{Im}_R G_2(z_1, \dots, z_m),$$
which happens if and only if $$G_1(z_1,\dots, z_m) = U(z_1,\dots,z_m)G_2(z_1,\dots, z_m)$$ for some unimodular matrix\index{Unimodular matrix} $U(z_1, \dots, z_m) \in R^{k \times k}$ \cite{Weiner1998}.\\


Let $\alpha = \begin{bmatrix} \alpha_1 & \dots & \alpha_m\end{bmatrix} \in \mathbb{N}^m$. By $z^{\alpha}$ we mean $z_1^{\alpha_1}\cdot \ldots \cdot z_m^{\alpha_m}$. A polynomial $f \in R$ can be written as $$f(z_1, \dots, z_m) = \sum\limits_{\alpha \in \mathbb{N}^{m}}f_{\alpha}z^{\alpha} \text{ where } f_{\alpha} \in \mathbb{F}_{q}.$$ In a similar way a vector $w = \begin{bmatrix} w_1 & \cdots & w_n \end{bmatrix} \in {R}^n$ can be written as $$w = \sum\limits_{\alpha \in \mathbb{N}^{m}}w_\alpha z^\alpha \text{ where } w_\alpha \in \mathbb{F}_{q}^{n}.$$

\begin{definition}
The \textbf{weight}\index{Weight} of $f \in R$ is denoted by $\operatorname{wt}(f)$ and is the number of nonzero terms of $f$. 
If $w = \begin{bmatrix} w_1 & \cdots & w_n \end{bmatrix} \in {R}^n$, then the weight of $w$ is given by $$\operatorname{wt}(w)=\sum\limits_{j=1}^{n} \operatorname{wt}(w_j).$$
Equivalently, if $w = \sum\limits_{\alpha \in \mathbb{N}^{m}}w_\alpha z^\alpha \text{ where } w_\alpha \in \mathbb{F}_{q}^{n}$ then $$\operatorname{wt}(w)=\sum\limits_{\alpha \in \mathbb{N}^m}\operatorname{wt}(w_\alpha).$$
\end{definition}


\begin{definition}[\hspace{-0.1mm}\cite{Weiner1998}] Given two elements $w, \Tilde{{w}} \in R^n$, the (Hamming) \textbf{distance} between $w$ and $\Tilde{w}$ is given by $d(w, \Tilde{w}) = \operatorname{wt}(w-\Tilde{w})$. 
The \textbf{free distance}\index{Free distance} of an $m$D convolutional code $\mathcal{C}\subset R^n$ is given by $$d_{free}(\mathcal{C})= \min \{ d(w, \Tilde{w}): w, \Tilde{w} \in \mathcal{C}, w \neq \Tilde{w}\}.$$
\end{definition}

The linearity of an $m$D convolutional code $\mathcal{C}$ implies that
$$d_{free}(\mathcal{C}) = \min \{\operatorname{wt}(w): w \in \mathcal{C}, w \neq 0\}.$$

The (total) degree of a polynomial $$f = \sum\limits_{\alpha \in \mathbb{N}^{m}}f_{\alpha}z^{\alpha} \in R$$ is  given by 
$$\max\{\alpha_1+ \dots+ \alpha_m: f_{\alpha} \neq 0\}.$$ 

The $i$-th row degree\index{Row degree}, $\nu_i$, of a polynomial matrix $G(z_1,\dots,z_m)\in R^m$ is defined as the maximum degree of the entries in its $i$-th row. The \textbf{external degree}\index{External degree} of $G(z_1, \dots, z_m)$ is the sum of its row degrees, i.e., $\sum\limits_{i=1}^{k} \nu_i$. 

\begin{definition}
Let $\mathcal{C}$ be an $m$D convolutional code. The \textbf{degree} of $\mathcal{C}$ is defined as the minimum of the external degrees among all the generator matrices of $\mathcal{C}$.    
\end{definition}


\begin{example}\label{newexampledegree2D}
Let $R = \mathbb{F}_2[z_1, z_2]$. Consider the $2$D convolutional code $\mathcal{C}$ with generator matrix
\[
G(z_1, z_2) = \begin{bmatrix}
1 & z_1 & 0 \\
1 & z_2 & 1
\end{bmatrix} \in R^{2 \times 3}.
\]
The external degree of $G_3(z_1, z_2)$ is $\nu_1+\nu_2 = 1 + 1 = 2$. We now show that no generator matrix for $\mathcal{C}$ can have external degree less than 2. Suppose, for contradiction, that there exists a generator matrix \( G'(z_1, z_2) \) of \( \mathcal{C} \) with one row degree equal to 0 (i.e., all entries in \( \mathbb{F}_2 \)).
 Let us denote it by $v \in \mathbb{F}_2^{1 \times 3}$. Then, since 
  \[
G'(z_1,z_2) = U(z_1,z_2) G(z_1,z_2),
\]
 for some unimodular matrix $U(z_1,z_2) \in R^{2 \times 2}$, we have that
\[
v = p \begin{bmatrix} 1 & z_1 & 0 \end{bmatrix} + q \begin{bmatrix} 1 & z_2 & 1 \end{bmatrix} 
= \begin{bmatrix} p + q & p z_1 + q z_2 & q \end{bmatrix},
\]
for some $p, q \in R$. Since $v \in \mathbb{F}_2^{1 \times 3}$, $p + q, q, pz_1+qz_2 \in \mathbb{F}_2$ which implies that $p = q = 0$, i.e., 
$v = \begin{bmatrix} 0 & 0 & 0 \end{bmatrix}$, which is not possible because a generator matrix must be full row rank. 
Hence, the minimum possible external degree for a generator matrix of $\mathcal{C}$ is 2 and then \( \mathcal{C} \) has degree 2.
\end{example}

The \textbf{internal degree}\index{Internal degree} of $G(z_1, \dots, z_m)\in R^{k\times n}$ is defined as the maximal degree of its full-size minors. Since the generator matrices of an $m$D convolutional code differ by pre-multiplication of a unimodular matrix, their full-size minors differ by a nonzero constant and therefore they have the same internal degree. We can therefore introduce the notion of complexity of $m$D convolutional codes.

\begin{definition}
Let $\mathcal{C}$ be an $m$D convolutional code. The \textbf{complexity}\index{Complexity} of $\mathcal{C}$ is defined as the internal degree of any generator matrix of $\mathcal{C}$.
\end{definition}

Note that the internal degree of a polynomial matrix is always smaller than or equal to its external degree. Thus the complexity of an $m$D convolutional code is smaller than or equal to its degree.

Observe also that the degree of a $1$D convolutional code is equal to its complexity, because a $1$D convolutional code always has row-reduced\index{Row-reduced} generator matrices where the external degree matches the internal degree. However, for $m$D convolutional codes, such generator matrices do not always exist and there exist $m$D convolutional codes where the complexity and degree can differ. 

\begin{example}
Let $R= \mathbb{F}_{2}[z_1, z_2]$. Consider again the $2$D convolutional code $\mathcal{C}$ with generator matrix
\[
G_{3}(z_1,z_2) = \begin{bmatrix}
1 & z_1 & 0\\
1 & z_2 & 1
\end{bmatrix} \in R^{2 \times 3}.
\]
Although $\mathcal{C}$ has degree 2 (see \hyperref[newexampledegree2D]{Example \ref{newexampledegree2D}}), its complexity is 1, since the $2 \times 2$ minors of $G_3(z_1,z_2)$ are:
\[
\begin{vmatrix}
1 & z_1 \\
1 & z_2
\end{vmatrix}
= z_1 + z_2, \quad
\begin{vmatrix}
1 & 0 \\
1 & 1
\end{vmatrix}
= 1, \quad
\begin{vmatrix}
z_1 & 0 \\
z_2 & 1
\end{vmatrix}
= z_1
\]
and the degrees of these minors are $1$, $0$, and $1$, respectively.
\end{example} 



The following theorem, presented in \cite{MDSMulti2025}, provides an upper bound on the distance of $m$D convolutional codes with rate $k/n$ and degree $\delta$, known as the \textbf{$m$D generalized Singleton bound}. For completeness, we include its proof below. Before doing so, we introduce a technical lemma that will be used in the proof of the theorem.

\begin{lemma}[\cite{10.5555/579402}]\label{card1}
  Let $m,\nu \in \mathbb N$ and 
  
  $$S=\{(i_1, \dots, i_m) : 0 \leq i_1+i_2+\cdots+ {i}_m \leq \nu\}.$$
  
    Then $\# S = \frac{(\nu + m)!}{\nu ! m!}$.
\end{lemma}

\begin{theorem}
\label{thmDgensinbou}
Let $\mathcal{C}$ be an $m$D convolutional code of rate $k/n$ and degree $\delta$. Then 
\begin{align*}
d_{free}(\mathcal{C})\leq n \frac{(\lfloor \frac{\delta}{k} \rfloor +m)!}{\lfloor \frac{\delta}{k} \rfloor!m!}-k\Big(\Big\lfloor \frac{\delta}{k} \Big\rfloor + 1\Big) + \delta +1.
\end{align*}
\end{theorem}
\begin{proof}
Let $G(z_1, \dots, z_m) \in R^{k \times n}$ be a generator matrix of $\mathcal{C}$ with row degrees $ \nu_1 \geq \nu_2 \geq \cdots\geq \nu_{t-1} > \nu_t = \nu_{t+1} = \nu_{t+2} = \cdots = \nu_k$, for some $t\in\{1,\dots k\}$, such that $\nu_1+\nu_2 +\dots + \nu_k = \delta$.

Let us write $$G(z_1, \dots, z_m) = 
\left[
\begin{array}{cccc}
G^{(1)}(z_1, \dots, z_m)\\
G^{(2)}(z_1, \dots, z_m)
\end{array}
\right]
$$ where $G^{(1)}(z_1, \dots, z_m) \in R^{(t-1) \times n}$ and $G^{(2)}(z_1, \dots, z_m) \in R^{(k-t+1) \times n}$. 

Note that $G^{(2)}(z_1, \dots, z_m)$ has row degrees all equal to $\nu_k$ and therefore we can write
$$G^{(2)}(z_1, \dots, z_m)=\sum\limits_{0 \leq \alpha_1+\cdots+\alpha_{{m}}\leq \nu_k}G^{(2)}_{\alpha_1 \cdots \alpha_{{m}}} z^{(\alpha_1, \cdots, \alpha_{{m}})}.$$

Since $G^{(2)}_{0,\dots,0}$ is a $(k-t+1)\times n$ matrix, there exists $\tilde{u} \in \mathbb{F}_{q}^{1\times (k-t+1)}\setminus\{0\}$ such that $\tilde{u}G^{(2)}_{0,\cdots,0}$ has $k-t$ zero entries. 

Let $u=\left[\begin{array}{cc}0_{1\times(t-1)} & \tilde{u}\end{array}\right]$. We have that
\begin{align*}
    \operatorname{wt}(uG(z_1, \dots, z_m)) & =  \operatorname{wt}(\tilde{u}G^{(2)}(z_1, \dots, z_m))\\
    & = \operatorname{wt}\big(\tilde{u}G^{(2)}_{0 \cdots 0}\big) + {\sum\limits_{1 \leq \alpha_1 + \dots + \alpha_{{m}} \leq \nu_k}\operatorname{wt}\big(\tilde{u}G^{(2)}_{\alpha_1 \cdots \alpha_{{m}}}\big)}
\end{align*} 
By \hyperref[card1]{Lemma \ref{card1}}, it follows that
\begin{align*}
    \operatorname{wt}(uG(z_1, \dots, z_m)) & \leq n-(k-t)+n \Big( \frac{(\nu_k+m)!}{\nu_k!m!} -1\Big)\\
    & = n \Big(\frac{(\nu_k+m)!}{\nu_k!m!}\Big) - (k-t)
\end{align*} 
and therefore $$d_{free} (\mathcal{C}) \leq n \frac{(\nu_k+m)!}{\nu_k!m!} - (k-t).$$ 
The maximum value of this upper bound is achieved by maximizing $\nu_k$ and then $t$, i.e. when $\nu_k = \Big\lfloor \frac{\delta}{k} \Big\rfloor$ and $t=\delta-k \Big\lfloor \frac{\delta}{k} \Big\rfloor+1$ and therefore
\begin{align*}
\label{eq1}
d_{free}(\mathcal{C})\leq n \frac{(\lfloor \frac{\delta}{k} \rfloor +m)!}{\lfloor \frac{\delta}{k} \rfloor!m!}-k\Big(\Big\lfloor \frac{\delta}{k} \Big\rfloor + 1\Big) + \delta +1.
\end{align*}
\end{proof}

An $m$D convolutional code with rate $k/n$ and degree $\delta$ whose free distance attains the corresponding $m$D generalized Singleton bound is called \textbf{Maximum Distance Separable (MDS) $m$D convolutional code}\index{Maximum Distance Separable (MDS)}.


\section{Construction of MDS $m$D Convolutional Codes of Rate \texorpdfstring{$1/n$}{1/n}}
\normalsize

In this section, we consider $m$D convolutional codes of rate $1/n$ and degree $\delta$, and we present constructions of such codes that are MDS.

The next definition defines superregular matrices which will be useful for the construction of these codes.


\begin{definition}\label{defsuperregular}
A matrix $A \in\mathbb F_q^{k\times n}$ is \textbf{superregular}\index{Superregular matrix} if all of its minors (of all sizes) are nonzero.
\end{definition}

It is clear that any submatrix of a superregular matrix is also superregular and that a matrix obtained from a superregular matrix by permutation of rows is also superregular.

\begin{lemma}\label{lemmalincob}
Let $A \in \mathbb F_q^{r \times s}$ with $r\leq s$ be a superregular matrix. Then, for any $u \in \mathbb{F}_{q}^{r}\setminus \{0\}$, we have that $$\operatorname{wt}(uA) \geq s-\operatorname{wt}(u)+1.$$
\end{lemma}



In \cite{LiebPinto2020} constructions of MDS 1D convolutional codes were presented using superregular matrices. Given a polynomial matrix $G(z)=\sum\limits_{i=0}^{\delta} G_iz^i \in \mathbb F[z]^{1 \times n}$, with $G_{\delta} \neq 0$, let us consider the $(\delta + 1) \times n$ constant matrix \begin{equation}\label{constmatrix}
\Phi_1(G(z))=\left[\begin{matrix}G_0 \\ G_1 \\ \vdots \\ G_{\delta}\end{matrix}\right].
\end{equation}

\begin{theorem}\cite{LiebPinto2020}\label{th1DMDS}
Let $n, \delta \in \mathbb N$ with $n \geq \delta + 1$ and $G(z)=\sum\limits_{i=0}^{\delta} G_iz^i \in \mathbb F[z]^{1 \times n}$, with $G_{\delta} \neq 0$. If $\Phi_1(G(z))$ is superregular, then $G(z)$ is a generator matrix of an MDS 1D convolutional code of rate $1/n$ and degree $\delta$.
\end{theorem}

We include the following theorem because it will be used later in the proof of the \hyperref[thmD]{Theorem \ref{thmD}}. It gives an upper bound on the free distance of 1D convolutional codes whose generator matrices have specific row degrees, and it also provides a condition under which such codes achieve that upper bound.

\begin{theorem}\label{th1D}
    Let $\ell\in\mathbb N_0$ and $G_i(z) \in\mathbb F_q[z]^{k_i\times n}$ with $k_i\in\mathbb N$ for $i=0,1,\dots, \ell-1$ and $G_{\ell}(z) \in\mathbb F_q[z]^{1\times n}$ such that $$G(z)=\left[\begin{array}{c} 
    G_0(z) \\
    G_1(z) \\
    \vdots \\ G_{\ell}(z) \end{array}\right] \in \mathbb F[z]^{(k_0+\cdots+k_{\ell-1}+1) \times n}$$ is the generator matrix of a convolutional code $\mathcal{C}$,
    where, for $i=0,1,\dots,\ell$, all row degrees of $G_i(z)$ are equal to $\nu+\ell-i$ for some $\nu \in \mathbb N$. 
    Then
    $$
    d_{free}({\cal C}) \leq n (\nu+1).
    $$
    Moreover, let ${\cal G}=\left[\begin{array}{c} 
   \Phi_1(G_0(z)) \\
    \Phi_1(G_1(z)) \\
    \vdots \\ \Phi_1(G_{\ell}(z)) \end{array}\right]\in\mathbb F^{L\times n}$ where $L=\displaystyle\sum_{j=0}^{\ell-1} k_j(\nu+\ell+1-j)+\nu+1$. If $n \geq L$ and ${\cal G}$ is superregular, then
    $$
    d_{free}({\cal C}) =  n (\nu+1).
    $$
\end{theorem}

\begin{proof}
First we observe that for $\ell=0$ the statement of this theorem is exactly the statement of \hyperref[th1DMDS]{Theorem \ref{th1DMDS}} and hence, for this proof we can assume $\ell\geq 1$.

Note that the codeword $G_{\ell}(z)$ has weight $n(\nu+1)$ which proves that $ d_{free}({\cal C}) \leq n (\nu+1)$.
To prove that $ d_{free}({\cal C}) = n (\nu+1)$ write
$$
G(z)=G^{(0)} + G^{(1)}z+ \cdots + G^{(\nu+1)}z^{\nu+1} + \cdots + G^{(\nu+\ell)}z^{\nu+\ell},
$$
where, for 
$j=0,1,\dots,\nu$, $G^{(j)}$ has $k_0+k_1+\cdots+k_{\ell-1}+1$ nonzero rows, and for $i=1,2,\dots,\ell$, 
\begin{align}\label{nonzero_rows_in_G}
G^{(\nu+i)}=\left[\begin{array}{c} \tilde G^{(\nu+i)} \\ 0 \end{array}\right]\quad\text{with}\ \ \  \tilde G^{(\nu+i)} \in \mathbb F_q^{k_0+k_1+\cdots + k_{\ell-i}},
\end{align} 
where all rows of $\tilde G^{(\nu+i)}$ are nonzero. Set $G^{(j)}=0$, for $j > \nu + \ell$.

Let us also write 
$$
u(z)=\sum_{i \in \mathbb N_0} u^{(i)}z^i,
$$
with $u^{(i)}=\left[\begin{array}{ccccc}  
u^{(i)}_0 &  u^{(i)}_1 & \cdots & u^{(i)}_{\ell-1} & u^{(i)}_{\ell}
\end{array}\right]$,
where $u^{(i)}_j \in \mathbb F_q^{k_j}$, $j=0,1,\dots,\ell-1$ and $u^{(i)}_{\ell} \in \mathbb F_q$. Set $u^{(j)}=0$, for $j <0$.

Let us prove that $w(z)=u(z)G(z)$ has weight greater or equal than $n(\nu+1)$. We can consider without loss of generality that $u^{(0)} \neq 0$. For \( w(z) = \sum\limits_{i \in \mathbb{N}_{0}} w_i z^{i} \), let us denote
$$ w|_{[r,s]} = \sum_{i=r}^{s} w_i z^{i}, \text{ for } ( r < s ).$$

Since $\cal{G}$ is a superregular matrix and $u^{(0)} \neq 0$ we have that $w(z)$ has degree greater or equal than $\nu$. \\
\underline{Case 1}: $\deg w(z)= \nu$.\\
Then,
\begin{align}\label{nu}
w^{(\nu+1)}= \left[\begin{array}{cccc}  
u^{(\nu+1)} &  u^{(\nu)} & \cdots & u^{(0)}
\end{array}\right] \left[\begin{array}{c}  
G^{(0)} \\  G^{(1)} \\ \vdots\\ G^{(\nu+1)}
\end{array}\right]=0.
\end{align}
Recall that for $j=0,\hdots,\nu$, all rows of $G^{(j)}$ are nonzero and $G^{(\nu+1)}$ has exactly the last row equal to zero since we can assume $\ell\geq 1$. After deleting this zero row, the matrix in \eqref{nu} is a submatrix of $\cal{G}$, which is full row rank since $\cal{G}$ is superregular with at least as many columns as rows. Hence, \eqref{nu} implies that $u^{(\nu+1)}=u^{(\nu)}=\cdots=u^{(1)}=0$ 
and $u^{(0)}_j=0$ for $j=0,1,\dots,\ell-1$.
Moreover, it can be easily seen that $u^{(\nu+i)}=0$ for $i \geq 2$,  i.e., $u(z)=\left[\begin{array}{cccc}  
0 & \cdots & 0 & \bar u
\end{array}\right]$ with $\bar u \in \mathbb F_q \backslash \{ 0 \}$. 
Thus, $w(z)=\bar u\, G_{\ell}(z)$, which has weight $n(\nu +1)$.\\
\underline{Case 2}: $\deg w(z)=\nu + i$ for some $i \geq 1$.\\
Then, $$
w^{(\nu+i+1)} = \left[\begin{array}{cccc}  
u^{(\nu+i+1)} &  u^{(\nu+i)} & \cdots & u^{(i+1-\ell)}
\end{array}\right] \left[\begin{array}{c}  
G^{(0)} \\  G^{(1)} \\ \vdots\\ G^{(\nu+\ell)}
\end{array}\right]=0,
$$
which implies that 
\begin{align}\label{zero_us}
u^{(\nu+i+1)}=u^{(\nu+i)}=\cdots= u^{(i+1)}=0.
\end{align}
Moreover, since $w^{(\nu+i+j)}=0$, for $j>2$, it follows that $u^{(\nu+i+j)}=0$, for $j>2$, which means that $\deg u(z)=i$. 

Furthermore, for $s=0, \dots, \ell-1$,
\begin{align}\label{nonzeros_in_u}
u^{(i-s)}= \left[\begin{array}{cc}  
0  & \bar u^{(i-s)}
\end{array}\right]\quad\text{where}\ \ \  \bar u^{(i-s)} \in \mathbb F_q^{1+k_{\ell-1}+\cdots+ k_{\ell-s}},
\end{align}
since $G^{(\nu+1+s)}$ has $k_0+k_1+\cdots+k_{\ell-1-s}$ nonzero rows; see \eqref{nonzero_rows_in_G}.

Then, for $j=0,1,\dots,\mbox{min}\{i-1, \nu\}$,
\begin{eqnarray*}w^{(\nu+i-j)}& =&\left[\begin{array}{cccccc}  
u^{(i)} & \cdots & u^{(i+1-\ell)} & u^{(i-\ell)} \cdots & u^{(i-\ell-j)}
\end{array}\right] \left[\begin{array}{c}  
G^{(\nu - j)} \\  \vdots \\ G^{(\nu+\ell)}
\end{array}\right] \\
& =&\left[\begin{array}{cccccc}  
\bar u^{(i)} & \cdots & \bar u^{(i+1-\ell)} & u^{(i-\ell)} \cdots & u^{(i-\ell-j)}
\end{array}\right] \left[\begin{array}{c}  
\bar G^{(\nu - j)} \\  \vdots \\ \bar G^{(\nu - j-1+\ell)} \\G^{(\nu - j+\ell)}\\ \vdots \\ G^{(\nu+\ell)}
\end{array}\right],
\end{eqnarray*}
where $G^{(\nu - j+s)}=\left[ \begin{array}{c}  
\hat G^{(\nu - j+s)}  \\ \bar G^{(\nu - j+s)} \end{array} \right]$ for $s=0,1,\dots, \ell-1$, 
with $\bar G^{(\nu - j)} \in \mathbb F_q^{1 \times n}$ 
and $\bar G^{(\nu - j+s)} \in \mathbb F_q^{(1+k_{\ell-1} + \cdots + k_{\ell-s}) \times n}$. 
Since the number of rows of the matrix $\left[\begin{array}{c}  
\bar G^{(\nu - j)} \\  \vdots \\ \bar G^{(\nu - j-1+\ell)} \\ G^{(\nu - j+\ell)}\\ \vdots \\ G^{(\nu+\ell)}
\end{array}\right]$ is 
\begin{align*}&\quad \sum_{s=0}^{\ell-1}(1+k_{\ell-1}+\cdots+k_{\ell-s})+\sum_{i=0}^j(k_0+\cdots+k_i)\\
&=\sum_{i=0}^{\ell-1}(1+k_{\ell-1}+\cdots+k_{i-1})+\sum_{i=0}^j(k_i+\cdots+k_0)\\
&=\begin{cases}(j+1)(k_0+k_1+\cdots+k_{\ell-1}+1)+\sum_{i=j+1}^{\ell-1}(1+k_{\ell-1}+\cdots+k_{i-1}), & j\leq \ell-1\\
\ell(k_0+k_1+\cdots+k_{\ell-1}+1)+\sum_{i=\ell}^j(k_i+\cdots+k_0) ,& j>\ell-1
\end{cases}
\end{align*}
where $k_{\ell}=1$ and $k_t=0$ for $t>\ell$, there are at least $(j+1)(k_0+k_1+\cdots+k_{\ell-1}+1)$ nonzero rows 
which are a submatrix of $\cal G$. Moreover, $\deg(u(z))=i$ implies $\bar u^{(i)} \neq 0$ and thus, by \hyperref[lemmalincob]{Lemma \ref{lemmalincob}}, we have that 
\begin{align}\label{j}
\mbox{wt}(w^{(\nu+i-j)})\geq n-(j+1)(k_0+k_1+\cdots+k_{\ell-1}+1)+1. 
\end{align}

Let us analyze the cases in which \( i \geq \nu+1 \) and \( i < \nu+1 \) separately.

\underline{Case 2.1}: $i \geq \nu+1$\\
In this case, \eqref{j} is true for $j=0,\hdots,\nu$ and we obtain {\small
\[
\begin{aligned}
\mbox{wt}(w_{|[i,\nu+i]}) &= \sum_{j=0}^\nu \mbox{wt}(w^{(\nu+i-j)})\\
&\geq n(\nu+1)-\frac{(\nu+1)(\nu+2)}{2}(k_0+k_1+\cdots+k_{\ell-1}+1) + \nu +1.
\end{aligned}
\]
}
Moreover, we have that for $t=0,\,\dots,\nu$,

\begin{align}\label{up_to_nu}
w^{(t)}=\left[\begin{array}{cccc}  
u^{(t)} & u^{(t-1)} & \cdots & u^{(0)}
\end{array}\right] \left[\begin{array}{c}  
G^{(0)} \\ G^{(1)} \\  \vdots \\ G^{(t)}
\end{array}\right]
\end{align}
has weight greater than or equal to 
$$
n-(t+1)(k_0+k_1+\cdots+k_{\ell-1}+1)+1,
$$
due to \hyperref[lemmalincob]{Lemma \ref{lemmalincob}}
and therefore,
$$\mbox{wt}(w_{|[0,\nu]}) \geq n(\nu+1)-\frac{(\nu+1)(\nu+2)}{2}(k_0+k_1+\cdots+k_{\ell-1}+1) + \nu +1.$$
In sum, one obtains \small{\begin{align*}
  & \mbox{wt}(w(z))\geq \mbox{wt}(w_{|[0,\nu]}) + \mbox{wt}(w_{|[i,\nu+i]})\\
   &\geq n(\nu+1)+n(\nu+1)-(\nu+1)(\nu+2)(k_0+k_1+\cdots+k_{\ell-1}+1)+2(\nu+1)\\
   &\geq n(\nu+1)+(\nu+1)(\nu+2)\sum_{j=0}^{\ell-1}k_j+(\nu+1)^2-(\nu+1)(\nu+2)\left(\sum_{j=0}^{\ell-1}k_j+1\right)+2(\nu+1)\\
   &\geq n(\nu+1)+(\nu+1)(\nu+1-(\nu+2)+2)\geq n(\nu+1),
\end{align*}}

where the third inequality follows from $$n \geq \sum_{j=0}^{\ell-1} k_j(\nu+\ell+1-j)+\nu+1\geq (\nu+2)\sum_{j=0}^{\ell-1}k_j+\nu+1.$$


\underline{Case 2.2}: $i < \nu+1$\\
In this case, we can use \eqref{j} for $j=0,\hdots,i-1$ to obtain $$\mbox{wt}(w_{|[\nu+1,\nu+i]}) \geq ni-\frac{i(i+1)}{2}(k_0+k_1+\cdots+k_{\ell-1}+1) + i.$$

Let us analyze the cases in which \( i \geq \ell-1 \) and $i<\ell-1$ separately.

\underline{Case 2.2.1}: $i < \nu+1$ and $i \geq \ell-1$\\
In this case, using again \hyperref[lemmalincob]{Lemma \ref{lemmalincob}}, we obtain
$$
\mbox{wt}(w^{(t)}) \geq n-(t+1)(k_0+k_1+\cdots+k_{\ell-1}+1)+1 \quad \text{for}\quad t=0,1,\dots, i-\ell,$$

For $t\geq i-\ell+1$, we can exploit that according to \eqref{nonzeros_in_u}, for $s=0,\hdots,\ell-1$, the vector $u^{(i-s)}$ has (at least) $k_0+\cdots+k_{\ell-s-1}$ zero components.
Hence, considering again \eqref{up_to_nu} and observing that for $t\geq i-\ell+1$, the vectors $u^{(t)},\hdots,u^{(i-\ell+1)}$ are involved in the calculation of $w^{(t)}$, we obtain

$$
\mbox{wt}(w^{(t)}) \geq n-(t+1)(k_0+k_1+\cdots+k_{\ell-1}+1)+1 + (\ell-i+t)k_0 + (\ell-i+t-1)k_1 + \cdots + k_{\ell-i+t-1},
$$
for $t=i-\ell+1,\dots, i$ and 
$$
\mbox{wt}(w^{(t)}) \geq n-(i+1)(k_0+k_1+\cdots+k_{\ell-1}+1)+1 + \ell k_0 + (\ell-1)k_1 + \cdots + k_{\ell-1},
$$
for $t=i, i+1, \hdots,\nu$. Therefore, 
$$\mbox{wt}(w_{|[0,i-1]}) \geq ni-\frac{i(i+1)}{2}(k_0+k_1+\cdots+k_{\ell-1}+1) + i + \sum_{j=0}^{\ell-1}\frac{(j+1)(j+2)}{2}k_{\ell-1-j}$$
and
{\small
\[
\begin{aligned}
\mbox{wt}(w_{|[i,\nu]}) &\geq n(\nu-i+1) - (\nu-i+1)(i+1)(k_0+k_1+\cdots+k_{\ell-1}+1) \\
&\quad + \nu - i +1 + (\nu-i+1)(\ell k_0 + (\ell-1)k_1 + \cdots + k_{\ell-1}).
\end{aligned}
\]
}

Thus, since $n\geq L= \displaystyle\sum_{j=0}^{\ell-1} k_j(\nu+\ell+1-j)+\nu+1$, we have that

\begin{align*}
&\mbox{wt}(w(z)) = \mbox{wt}(w_{|[0,\nu]}) + \mbox{wt}(w_{|[\nu+1,\nu+i]})\\
&\geq n(\nu+1)+ni-(i+1)(\nu+1)\left(\sum_{j=0}^{\ell-1}k_0+1\right)+\nu+1+i+\\
&+\sum_{j=0}^{\ell-1}\left(\frac{(j+1)(j+2)}{2}+(j+1)(\nu-i+1)\right)k_{\ell-1-j}\\
&\geq n(\nu+1)+\\
&+ \sum_{j=0}^{\ell-1}\left(i(\nu+\ell+1-j)-(i+1)(\nu+1)+\frac{(j+1)(j+2)}{2}+(j+1)(\nu-i+1)\right)k_{\ell-1-j}+\\
&+i(\nu+1)-(i+1)(\nu+1)+\nu+1+i\\
&\geq n(\nu+1)+\sum_{j=0}^{\ell-1}\left(i(\ell-j-1)+\frac{(j+1)(j+2)}{2}+j(\nu-i+1)\right)k_{\ell-1-j}\geq n(\nu+1).
\end{align*}

\underline{Case 2.2.2: $i < \nu+1$ and $i < \ell-1$}\\
In this case, we have that
$$
w^{(\nu+i+1)} = \left[
\begin{array}{ccc}
u^{(i)} & \cdots & u^{(0)}
\end{array}
\right]
\left[
\begin{array}{c}
G^{(\nu+1)} \\ \vdots \\ G^{(\nu+i+1)}
\end{array}
\right]=0
$$
implies that $u(z)=\left[
\begin{array}{cccc}
0 & \cdots & 0 &  \hat u(z)
\end{array}
\right]$ with $\hat u(z) \in \mathbb F_q[z]^{k_{\ell-i} + \cdots + k_{\ell-1}+1}$, and therefore $w(z)= \hat u(z) \left[
\begin{array}{c}
G_{\ell-i}(z) \\ \vdots \\ G_{\ell}(z)
\end{array}
\right]$.
Setting $\hat{\ell}=i$ and $\hat{k}_j=k_{j+\ell-i}$ for $j=0,\hdots,i$, i.e. in particular $\hat{\ell}\geq i$, brings us into the situation of Case 2.2.1. Hence, we obtain $wt(w(z)\geq n(\nu+1)$.
\end{proof}

Let us now consider general $m$-dimensional convolutional codes. Next, we introduce the following lemma, which will be used in the proof of the subsequent theorem.

 \begin{lemma}\label{card2}
     Let $m,\nu \in \mathbb N$. Then
     
     $$
     \frac{(\nu+m)!}{\nu!m!}= \sum_{i=0}^{\nu}\frac{(\nu+m-i-1)!}{(\nu-i)!(m-1)!}
     $$
 \end{lemma}

 \begin{proof}
Let $S=\{(i_1, \dots, i_m) : 0 \leq i_1+i_2+\cdots+ i_m \leq \nu\}$. By Lemma \ref{card1}, $\# S = \frac{(\nu + m)!}{\nu ! m!}$. On the other hand we can write $S$ as the disjoint union of sets as follows
$$
S=S_0 \cup S_1 \cup \cdots \cup S_{\nu},
$$
where $S_j=\{(i_1, \dots, i_{m-1},j) \, : \, 0 \leq i_1+i_2+\cdots+ i_{m-1} \leq \nu -j\}$, $j=0,1,\dots, \nu$. By Lemma \ref{card1}, $\# S_j=\frac{(\nu-j+m-1)!}{(\nu-j)!(m-1)!}$, and the result follows.
 \end{proof}
 
Given a polynomial matrix $G(z_1, \dots, z_m)=\sum\limits_{i=0}^{\mu} G_i(z_1, \dots, z_{m-1})z_m^i \in R^{1 \times n}$, with $G_{\mu}(z_1, \dots, z_{m-1}) \neq 0$, let us define the constant matrix $$\Phi_m(G(z_1, \dots, z_{m}))=\left[\begin{matrix} \Phi_{m-1}(G_0(z_1, \dots, z_{m-1})) \\\Phi_{m-1}(G_1(z_1, \dots, z_{m-1})) \\ \vdots \\ \Phi_{m-1}(G_{\mu}(z_1, \dots, z_{m-1}))\end{matrix}\right], \mbox{ for } m \geq 2,$$
with $\Phi_1(G_i(z_1))$ as defined in (\ref{constmatrix}).

\begin{theorem}\label{thmD}
    Let $$G(z_1, \dots,z_m)=\left[\begin{array}{c}
    G_0(z_1, \dots, z_m)\\
    G_1(z_1, \dots, z_m) \\
    \vdots \\ G_{\ell}(z_1, \dots, z_m) \end{array}\right] \in R^{(k_0+\cdots+k_{\ell-1}+1) \times n}$$ be a generator matrix of an mD convolutional code $\mathcal{C}$, with\\ $G_0(z_1, \dots, z_m) \in \mathbb{F}_{q}[z_1,\dots,z_m]^{k_0 \times n}$, 
    $G_1(z_1, \dots, z_m) \in \mathbb{F}_{q}[z_1,\dots,z_m]^{k_1 \times n},$ $\dots$,\\ $G_{\ell-1}(z_1, \dots, z_m) \in \mathbb{F}_{q}[z_1,\dots,z_m]^{{k}_{\ell-1} \times n}$, $ G_{\ell}(z_1, \dots, z_m) \in \mathbb{F}_{q}[z_1,\dots,z_m]^{1 \times n}$ of degrees $\nu_0 \geq \nu_1 \geq \cdots \geq \nu_{\ell-1} > \nu_{\ell}$, respectively. Then,
    $$
    d_{free}({\cal C}) \leq n \frac{(\nu_{\ell}+m)!}{\nu_{\ell}!m!}.
    $$
    Moreover, if $n \geq k_0(\nu_0+1)+k_1(\nu_1+1)+\cdots+k_{\ell-1}(\nu_{\ell-1+1})+\nu_{\ell}+1$ and ${\Phi_{m}(G(z_1,\dots, z_m))}$ is superregular, then
    $$
    d_{free}({\cal C}) = n \frac{(\nu_{\ell}+m)!}{\nu_{\ell}!m!}.
    $$
\end{theorem}

\begin{proof}
 Consider the codeword
    \begin{eqnarray*}
   w(z_1,\dots, z_m)& = &\left[\begin{array}{cccc} 
   0 & \cdots & 0 & 1
  \end{array}\right]
  \left[\begin{array}{c} 
    G_0(z_1,\dots, z_m) \\
    G_1(z_1,\dots, z_m) \\
    \vdots \\ G_{\ell}(z_1,\dots, z_m) \end{array}\right]\\
    & = & G_{\ell}(z_1,\dots, z_m). 
   \end{eqnarray*}
   Then, by \hyperref[card1]{Lemma \ref{card1}}
   $$
   d_{free}({\cal C}) \leq \operatorname{wt}(G_{\ell}(z_1,\dots, z_m)) \leq n\cdot  \# Supp(G_{\ell}(z_1,\dots, z_m)) = n \frac{(\nu_{\ell}+m)!}{\nu_{\ell}!m!},
   $$

   where the support of $ G(z_1,\dots, z_m)=\sum\limits_{i_1,\dots, i_m \in \mathbb N} G_{i_1,\dots,i_m}z_1^{i_1}z_2^{i_2}\cdots z_m^{i_m} \in R^m$ is defined as $$Supp(G(z_1, \dots, z_m))=\{(i_1,\dots,i_m) \in \mathbb N^{m} \,: \, G_{i_1,\dots,i_m} \neq 0\}.$$


   Let us assume that ${\Phi_{m}(G(z_1,\dots, z_m))}$ is superregular and let us prove that \begin{equation}\label{dfree} d_{free}({\cal C}) = n \frac{(\nu_{\ell}+m)!}{\nu_{\ell}!m!}\end{equation}
   by induction. By \hyperref[th1D]{Theorem \ref{th1D}}, (\ref{dfree}) is valid for $m=1$. Let us now consider that (\ref{dfree}) is valid for $m-1$, i.e., let us consider as induction hypothesis that if $$\tilde G(z_1, \dots,z_{m-1})=\left[\begin{array}{c}
    \tilde G_0(z_1, \dots, z_{m-1}) \\
    \tilde G_1(z_1, \dots, z_{m-1}) \\
    \vdots \\ \tilde G_{\tilde \ell}(z_1, \dots, z_{m-1}) \end{array}\right] \in \mathbb F_q[z_1,\dots, z_{m-1}]^{(\tilde \ell+1) \times n}$$ is a generator matrix of an (m-1)D convolutional code $\tilde {\mathcal{C}}$, with $\tilde G_0(z_1, \dots, z_{m-1})$,\\ $\tilde G_1(z_1, \dots, z_{m-1})$, $\dots$,  $\tilde  G_{\tilde \ell}(z_1, \dots, z_{m-1})$ of degrees $\tilde \nu_0 \geq \tilde \nu_1 \geq \cdots \geq \tilde \nu_{\ell-1} > \tilde \nu_{\ell}$, respectively,  and ${\Phi_{m-1}(\tilde G(z_1,\dots, z_m))}$ is superregular, then
    $$
    d_{free}(\tilde {\cal C}) = n \frac{(\nu_{\tilde \ell}+m-1)!}{\nu_{\tilde \ell}!(m-1)!}.
    $$

  Write
   \begin{eqnarray*}
   G(z_1,\dots, z_m) & = & G^{(0)}(z_1,\dots, z_{m-1}) + G^{(1)}(z_1,\dots, z_{m-1})z_m
 + \cdots + \\
 & & \;\; +G^{(\nu_{\ell})}(z_1,\dots, z_{m-1})z_m^{\nu_{\ell}} + \cdots + G^{(\nu_{0})}(z_1,\dots, z_{m-1})z_m^{\nu_{0}}.
 \end{eqnarray*}

 Note that $G^{(i)}(z_1,\dots, z_{m-1})$ has row degrees $\nu_0-i, \nu_1-i, \dots, \nu_{\ell}-i$, for $i=0, \dots, \nu_{\ell}$.

 Let us consider $u(z_1,\dots, z_m) \in R^{\ell+1} \backslash \{0\}$ and write
 $$
 u(z_1,\dots, z_m)=u^{(0)}(z_1,\dots, z_{m-1}) + u^{(1)}(z_1,\dots, z_{m-1})z_m
 + \cdots + u^{(t)}(z_1,\dots, z_{m-1})z_m^{t},
 $$
 for some $t \in \mathbb N$. Let
 $$
 w(z_1,\dots, z_m)=u(z_1,\dots, z_m)G(z_1,\dots, z_m), $$
 and let us prove that 
\[
\operatorname{wt}(w(z_1, \dots, z_m)) 
\geq 
n \, \frac{(\nu_{\ell} + m)!}{\nu_{\ell}! \, m!}.
\]
We can consider, without loss of generality, that 
\( u^{(0)}(z_1, \dots, z_{m-1}) \neq 0 \).
Write
\[
w(z_1, \dots, z_m)
= \sum\limits_{i \in \mathbb{N}}
w^{(i)}(z_1, \dots, z_{m-1}) \, z_m^{i},
\]
where
\[
\begin{aligned}
w^{(i)}(z_1, \dots, z_{m-1})
&=
\begin{bmatrix}
    u^{(i)}(z_1, \dots, z_{m-1}) &
    \cdots &
    u^{(0)}(z_1, \dots, z_{m-1})
\end{bmatrix} \\[4pt]
&\quad \times
\begin{bmatrix}
    G^{(0)}(z_1, \dots, z_{m-1}) \\[3pt]
    \vdots \\[3pt]
    G^{(i)}(z_1, \dots, z_{m-1})
\end{bmatrix}
z_m^{i},
\qquad \text{for } i = 0, 1, \dots, \nu_{\ell}.
\end{aligned}
\]

Note that since $\Phi_m(G(z_1,\cdots,z_m))$ is a superregular matrix, 
all the matrices 
\[
\begin{aligned}
&\Phi_{m-1}\!\left(\begin{bmatrix}
    G^{(0)}(z_1, \dots, z_{m-1})
    \end{bmatrix}
\right)
\\[1em]
&\Phi_{m-1}\!\left(
\begin{bmatrix}
    G^{(0)}(z_1, \dots, z_{m-1}) \\
    G^{(1)}(z_1, \dots, z_{m-1})
\end{bmatrix}
\right)
\\[1em]
& \qquad\qquad\qquad\qquad\quad\vdots
\\[1em]
&\Phi_{m-1}\!\left(
\begin{bmatrix}
    G^{(0)}(z_1, \dots, z_{m-1}) \\
    \vdots \\
    G^{(\nu_{\ell})}(z_1, \dots, z_{m-1})
\end{bmatrix}
\right)
\end{aligned}
\]

are also superregular. Moreover, the matrices
$$
 G^{(0)}(z_1,\dots, z_{m-1}), \left[\begin{array}{c}
    G^{(0)}(z_1,\dots, z_{m-1})  \\G^{(1)}(z_1,\dots, z_{m-1}) 
 \end{array}\right], \cdots, \left[\begin{array}{c}
    G^{(0)}(z_1,\dots, z_{m-1})  \\ \vdots \\ G^{(\nu_{\ell})}(z_1,\dots, z_{m-1}) 
 \end{array}\right]
 $$
 have the last row of degree $\nu_{\ell},\nu_{\ell}-1, \dots, 0$, respectively, and the remaining rows of each matrix have higher degree than their last row. Thus, by induction hypothesis, 
 $$
 \operatorname{wt}(w^{(i)}(z_1,\dots, z_{m-1})) \geq  n \frac{(\nu_{\ell}+m-i-1)!}{(\nu_{\ell}-i)!(m-1)!},
 $$
 $i=0,1,\dots, \nu_{\ell}$.
 Thus, by Lemma \ref{card2}, $\operatorname{wt}(w(z_1, \dots,z_m)) \geq    n  \frac{(\nu_{\ell}+m)!}{\nu_{\ell}!m!}.$
 
 \end{proof}

 A direct consequence of this theorem is the following condition for the existence of MDS $m$D convolutional codes of rate $1/n$, since for $k=1$, the generalized Singleton bound from \hyperref[thmDgensinbou]{Theorem \ref{thmDgensinbou}} is equal to $n  \frac{(\delta+m)!}{\delta!m!}$.

 \begin{theorem}
     Let $G(z_1, \dots,z_m) \in R^{1 \times n}$ with row degree equal to $\delta$ and $n \geq \delta+1$, such that $\Phi_m(G(z_1, \dots,z_m))$ is superregular. Then $G(z_1, \dots,z_m)$ is a generator matrix of an MDS $m$D convolutional code of rate $1/n$ and degree $\delta$.
 \end{theorem}


More general, \hyperref[thmD]{Theorem \ref{thmD}} gives conditions for construction of MDS $m$D convolutional codes of rate $k/n$ and degree $\delta=k \nu+k-1$, for some $\nu \in \mathbb N$, as stated in the following theorem. Note that an MDS convolutional code $\cal C$ with these parameters has free distance equal to $ n  \frac{(\nu+m)!}{\nu!m!}$, \hyperref[thmDgensinbou]{Theorem \ref{thmDgensinbou}}, and therefore any generator matrix of such code can not have a row with degree smaller that $\nu$ since such row would be a codeword of $\cal C$ with weight smaller than  $ n  \frac{(\nu+m)!}{\nu!m!}$.

\begin{theorem}
   Let $n \geq \delta+k=k(\nu+2)-1$. Let $$G(z_1, \dots,z_m)=\left[\begin{array}{c}
    G_0(z_1, \dots, z_m) \\
    G_1(z_1, \dots, z_m) \\
    \vdots \\ G_{k-1}(z_1, \dots, z_m) \end{array}\right] \in R^{k \times n}$$ be a generator matrix of an mD convolutional code $\mathcal{C}$, with $G_i(z_1, \dots, z_m)$ of degree $\nu + 1$, $i=0,1,\dots, k-2$, and  $G_{k-1}(z_1, \dots, z_m)$ of degree $\nu$, such that $\Phi_m(G(z_1, \dots,z_m))$ is superregular. Then $G(z_1, \dots,z_m)$ is a generator matrix of an MDS $m$D convolutional code of rate $k/n$ and degree $\delta=k \nu+k-1$.
\end{theorem}

\section{Conclusions}

The results presented in this paper help to improve the understanding of MDS $m$D convolutional codes and point to several possible directions for future work, such as exploring generator matrices with different row degree conditions, considering different superregular structures, and finding explicit constructions over small finite fields.

\section*{Acknowledgements}
This work is supported by The Center for Research and Development in Mathematics and Applications (CIDMA) under the Portuguese Foundation for Science and Technology 
(FCT, \url{https://ror.org/00snfqn58})  
Multi-Annual Financing Program for R\&D Units,
grants UID/4106/2025 and UID/PRR/4106/2025.
The second author is supported by the German research foundation, project number 513811367.

\bibliographystyle{alpha}
\bibliography{mult}

\end{document}